\documentclass[12pt]{article}
\usepackage[latin1]{inputenc}
\usepackage{fullpage}
\usepackage{amsfonts}
\usepackage{latexsym,amssymb}
\usepackage{amsmath}
\usepackage{pstricks,pst-plot,pst-node,pst-text,pst-3d,pst-tree}
\usepackage{natbib}
\usepackage{euscript}

\title{Preferences Yielding the ``Precautionary Effect''\\
}

\author{Michel \textsc{De Lara}\\
Universit\'e Paris-Est, Cermics,\\
6-8 avenue Blaise Pascal, 77455 Marne la Vall\'ee Cedex 2, France}
 
 \newtheorem{theorem}{Theorem}
 \newtheorem{proposition}[theorem]{Proposition}
 
 \newtheorem{corollary}[theorem]{Corollary}

\newenvironment{proof}{\small{\bf Proof.}}{\hfill$\Box$\normalsize
\bigskip}

\def\llower{_0}
\def\upper{_1}

\def\BB{{\mathbb B}}

\newcommand{\RR}{{\mathbb R}} 
\newcommand{\EE}{{\mathbb E}} 
\newcommand{\YY}{{\mathbb Y}} 
\newcommand{\PP}{{\mathbb P}} 

\newcommand{\Om}{\Omega}
\newcommand{\om}{\omega}

\def\calP{\mathcal{P}}

\def\text#1{\quad\mbox{#1}\quad} 
\def\mtext#1{\,\mbox{#1}\,} 
\def\defegal{:=}

\def\util{U}
\def\one{a}

\def\two{b} 
\def\TWO{\mathbb{B}} 
\def\dec{\TWO}

\def\condexpect#1#2#3{\EE_{#1}[ #2 \mid #3 ]}
\def\expect#1#2{\EE_{#1}[ #2 ]}
\def\Expect#1#2{\EE_{#1}\big[ #2 \big]}

\def\signal{y}
\def\Signal{Y}
\def\SIGNAL{\YY}
\def\rv{x}
\def\Rv{X}
\def\RV{{\mathbb X}}

\def\valeur{\mathbb V}

\def\info{\Delta \! \valeur}
\def\argmax{\mathop{\rm arg\,max}}
\def\argmin{\mathop{\rm arg\,min}}
\def\interval{{\mathbb I}}

\def\Matrix{M}

\def\mathscr{\EuScript}
\newcommand{\tribu}[1]{\mathscr{#1}}

\begin{document}
\maketitle
 

\def\aftertwo{%
\psTree{\Tdia{alea} \naput{\red second decision $\two$}}
\Toval{} \Toval{} \Toval{}
\endpsTree
}

\def\scheme{%
 \psTree[treemode=R]{\Tp}
\psTree{\Tn}
\psTree{\Tdia{\red signal} \naput{\red first decision $\one$}}
\psTree[treefit=loose,levelsep=*3cm]{\Tcircle[radius=1]{High}}
\aftertwo
\endpsTree
\psTree[treefit=loose,levelsep=*3cm]{\Tcircle[radius=1]{Low}}
\aftertwo
\endpsTree
\endpsTree
\endpsTree
\endpsTree
}


\def\figone{%
\psovalbox[linecolor=blue,shadow=true,fillstyle=solid,fillcolor=green]
{\begin{tabular}{c}
First decision \\ \red $\one \in \interval \subset \RR $
\end{tabular} }
}

\def\figtwo{%
\psovalbox[linecolor=blue,shadow=true,fillstyle=solid,fillcolor=green]
{\begin{tabular}{c}
Second decision \\ \red  $\two \in \dec(\one) \subset \BB$ 
\end{tabular} }
}

\def\figrv{%
\psovalbox[linecolor=blue,shadow=true,fillstyle=solid,fillcolor=green]
{\begin{tabular}{c}
Uncertainty \\ \red $\rv \in \{\rv_1, \ldots, \rv_m \}$ 
\end{tabular} } 
}

\def\figom{%
\pscirclebox[linecolor=blue,shadow=true]
{\begin{tabular}{c}
Sample space \\ \red $\om \in \Om$
\end{tabular} }
}

\def\figsignal{%
\psovalbox[linecolor=blue,shadow=true]
{\begin{tabular}{c}
Signal \\ \red $\signal \in \{\signal_1, \ldots, \signal_n \}$ 
\end{tabular} } 
}

\def\figutil{%
\psshadowbox[fillstyle=solid,fillcolor=yellow]
{\begin{tabular}{c}
Utility \\ \red $\util(\one,\two,\rv)$ 
\end{tabular} }
}

\def\model{%
\psset{unit=0.5cm,nodesep=1mm}
\rnode{one}{\figone} 
\rput(6,-7){\rnode{util}{\figutil}}
\rnode{two}{\figtwo} 
\rput(-6,10){\rnode{signal}{\figsignal}}
\rnode{rv}{\figrv}
\rput(0,10){\rnode{om}{\figom}}

\ncarc[linecolor=black]{->}{one}{util}
\ncarc[linecolor=black]{->}{rv}{util}
\ncarc[linecolor=black]{->}{two}{util}
\ncarc[linecolor=cyan]{->}{om}{rv}
\ncput*{Random variable $\Rv$}
\ncarc[linecolor=cyan]{->}{om}{signal}
\ncput*{Signal $\Signal$}
\ncarc[linecolor=blue]{->}{signal}{two}
\ncput*{Information}
}


\begin{abstract}
 Consider an agent taking two successive decisions to maximize his
expected utility under uncertainty.  
After his first decision, a signal is revealed that provides information
about the state of nature. 
The observation of the signal allows the decision-maker to revise his
prior and the second decision is taken accordingly. 
Assuming that the first decision is a scalar representing consumption,
the \emph{precautionary effect} holds when initial consumption is less
in the prospect of future information than without (no signal).
\citeauthor{Epstein1980:decision} in \citep*{Epstein1980:decision} 
has provided the most operative tool to exhibit the 
precautionary effect. Epstein's Theorem holds true when the difference
of two convex functions is either convex or concave, which is not a
straightforward property, and which 
is difficult to connect to the primitives of the economic model.
Our main contribution consists in giving a geometric characterization 
of when the difference of two convex functions is convex,
then in relating this to  the primitive utility model.
With this tool, we are able to study and unite a large body of the
literature on the precautionary effect.
\end{abstract}

\begin{quote}

\end{quote}

\section{Introduction}

Consider an agent taking two successive decisions to maximize his
expected utility under uncertainty.  
As illustrated in Figure~\ref{fig:scheme},
\begin{figure}
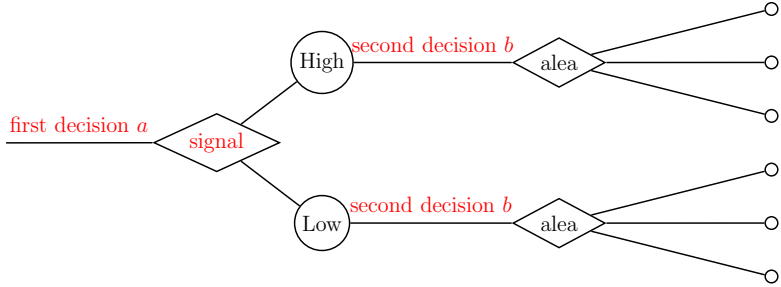

\psscalebox{0.7}{\scheme}
\caption{Decision  with learning; agent takes decision $\one$; 
a signal is revealed; agent takes decision $\two$ accordingly.}
\label{fig:scheme}
\end{figure}
after his first decision, a signal is revealed that provides information
about the state of nature. 
The observation of the signal allows the decision-maker to revise his
prior and the second decision is taken accordingly. 
Assuming that the first decision is a scalar representing consumption,
the \emph{precautionary effect} holds when initial consumption is less
in the prospect of future information than without (no signal).

The example above is a stereotype of sequential decisions problems with
learning were focus is put on comparison of the optimal initial decisions with
different information structures. For instance, should we aim at more
reductions of current greenhouse gases emissions today 
whether or not we assume some future improvement of our 
 knowledge about the climate?
Economic analysis has identified effects that go in opposite
directions and make the conclusion elusive. 
This article proposes a characterization of utility functions such that 
the precautionary effect holds for all signals.

Seminal literature in environmental economics
(\citep*{Arrow.ea1974:environmental},
\citep*{Henry1974:investment,Henry1974:option}) 
 focused on the irreversible environmental consequences carried by the initial
 decision and showed that the possiblity of learning should
lead to less irreversible  current decisions (``irreversibility effect'').
\citeauthor{Henry1974:investment, Arrow.ea1974:environmental}
consider additive separable
preferences, and so do \citep*{Freixas.ea1984:irreversibility}, 
\citep*{Fisher.ea1987:quasi}, 
\citep*{Hanemann1989:information}. 
\citeauthor{Epstein1980:decision} 
in \citep*{Epstein1980:decision} studies a more general
nonseparable expected utility model, and derives a condition that
identifies the direction of the precautionary effect.
His contribution remains the most operative tool.
Yet, the conditions under which this  result  holds
are difficult to connect to the primitive utility model.
Further contributions have insisted on the
existence of an opposite  economic irreversibility since environmental
precaution imply sunk costs that may constrain future consumption
 (\citep*{Kolstad1996:fundamental},
\citep*{Pindyck2000:irreversibilities}, 
\citep*{Fisher.ea2003:global}).
Risk neutral preferences are studied in \citep*{Ulph.ea1997:global}
for a global warming model. 
Assuming time separability of preferences, the papers 
\citep*{Gollier.ea2000:scientific} and
\citep*{Eeckhoudt-Gollier-Treich:2005} 
examine risk averse preferences.
\citeauthor*{Gollier.ea2000:scientific} identified  conditions on the
second-period utility function
 for the possibility of learning to have a precautionary
effect with and alternatively without the irreversibility constraint.
By this latter, we mean that the domain of the second decision variable
depends on the first decision.

The driving idea for linking the effect of learning and the value of
information is the observation that, once an initial decision is made,
 the value of information  can be defined as a \emph{function} of that
 decision. This is the approach of \citeauthor{Jones-Ostroy:1984} 
who define the value of information in this
way in their paper \citep*{Jones-Ostroy:1984} where they 
formalize the notion of
flexibility in a sequential decision context, and relate its value to
the amount of information an agent expects to receive.
Whenever the \emph{second-period value of information} -- namely 
the value of information  measured after an initial commitment is made
-- is a monotone function of the initial decision,
optimal initial decisions can be ranked. 
All this is recalled in
Sect.~\ref{sec:The_precautionary_effect:_statement_and_recalls},
where we also extend the approach in \citep*{Epstein1980:decision} and
\citep*{Jones-Ostroy:1984} to non necessarily finite sets.

In \citep*{Jones-Ostroy:1984}, the monotonicity property of
the second-period value of information is related to
convexity in the prior of a difference of maximal payoffs.
This is convexity is far from being granted since the difference of two
convex functions is generally not convex. 
Our main result consist in giving a general condition under which 
\emph{a difference of maximal payoffs exhibits convexity in the prior}. 
This is the object of Sect.~\ref{sec:primitive} where, first, we provide
a geometric characterization and, second, carry it to the utility model.
With this tool, we are able to study in Sect.~\ref{sec:Analysis of
  examples in the literature} a large body of the literature on
the precautionary effect.

\section{The precautionary effect: statement and recalls}
\label{sec:The_precautionary_effect:_statement_and_recalls}

We first give a formal statement of the precautionary effect, then
recall and extend some results in the literature, upon which we shall
elaborate our main contribution in the next section.

We shall assume that all sets are Borel spaces, 
endowed with the Borel $\sigma$-algebra generated by the
open subsets, that 
all mappings have appropriate measurability and integrability
properties needed to perform mathematical expectations operations,
and that all probabilities are regular to ensure the existence of
conditional distributions \citep*{Bertsekas-Shreve:1996},
\citep*{Kallenberg:2002}. 

\subsection{Problem statement}

Consider an agent taking two successive decisions as in
Figure~\ref{fig:scheme}.
The initial decision $\one$ is a scalar belonging to an interval
$\interval$ of $\RR$; 
the following and final decision $\two$ belongs to a set 
$\dec(\one)$ which may depend on the initial decision\footnote{%
This may materialize ``irreversibility'' of the initial decision.}
$\one$ ($\dec(\one)$ is a subset of a fixed set $\BB$).
Uncertainty is represented by states of nature $\om \in \Om$ with
prior $\PP$ on the Borel $\sigma$-field $\tribu{F}$, 
and by a random variable $\Rv : (\Om,\tribu{F}) \to (\RV,\tribu{\Rv})$.
Partial information on $\Rv$ is provided by means of a 
signal (random variable) 
$\Signal : (\Om,\tribu{F}) \to (\SIGNAL,\tribu{\Signal})$.
A utility function $\util(\one,\two,\rv)$ is given, defined on 
$ \interval \times \BB \times \Om$.
The expected utility maximizer solves\footnote{%
We shall always assume that, for the problems we  consider, the
\emph{sup} is attained and we shall use the notation \emph{max}. 
}
\begin{equation}
\max_{\one \in \interval} \Expect{}{ \max_{\two \in \dec(\one)}
\condexpect{}{ \util\big( \one,\two,\Rv \big) }{\Signal} } \; .
\end{equation}
Thus, the second decision $\two$ is taken knowing $\Signal$. 

The evaluation of expected utility right after the first decision $\one$
has been taken is conditional on the signal $\Signal$
and defined as follows:
\begin{equation}
\valeur^\Signal(\one) \defegal \Expect{}{ \max_{\two \in \dec(\one)}
\condexpect{}{ \util\big( \one,\two,\Rv \big) }{\Signal} } \; . 
\label{eq:valuesignal}
\end{equation}
With this notation, the program of the $\Signal$-informed agent is 
$\max_{\one} \valeur^\Signal(\one)$. 
Let us assume, for the sake of simplicity, 
that an optimal solution exists and is unique, 
denoted by $\bar{\one}^\Signal$.

A signal $\Signal$ is said to be \emph{more informative} than a
signal $\Signal'$ if the $\sigma$-field 
$\sigma(\Signal) \defegal \Signal^{-1}(\tribu{\Signal})$ contains 
$\sigma(\Signal')$. It is equivalent to say that $\Signal'$ is a
measurable function of $\Signal$, namely $\Signal'=f(\Signal)$ where 
$f : (\SIGNAL,\tribu{\Signal}) \to (\SIGNAL',\tribu{\Signal}') $.

The \emph{precautionary effect} is said to hold whenever 
the optimal initial decision is lower with more information that is,
if when the signal $\Signal$ is more informative than the signal $\Signal'$,
then $\bar{\one}^\Signal \leq \bar{\one}^{\Signal'}$. 

\subsection{Precautionary effect and second-period value of the
  information monotonicity} 

We shall recall a sufficient condition under which 
the presence of learning affects the first optimal decision in a
predictable way (see \citep*{Jones-Ostroy:1984}, 
\citep*[p.229]{DeLara-Doyen:2008}, \citep*{DeLara-Gilotte:2009}). 

Let us compare the programs of the $\Signal$-informed and
$\Signal'$-informed agent by writing
\begin{equation}
\max_{\one} \valeur^\Signal(\one) = 
 \max_{\one} \{ \valeur^{\Signal'}(\one) +
\big( \valeur^\Signal(\one) - \valeur^{\Signal'}(\one) \big) \} \; .
\end{equation}
It appears that the decision maker who expects more information
optimizes the same objective as the less informed decision maker \emph{plus} 
what we shall coin the \emph{second-period value of information} 
\begin{equation}
 \info^{\Signal \Signal'}(\one) \defegal 
\valeur^\Signal(\one) - \valeur^{\Signal'}(\one) 
\end{equation}
which depends on his initial decision. The more-informed agent 
initial optimal decision achieves a trade-off: 
it can be suboptimal from the point of
view of the less-informed decision maker but compensates for this by  an
increase of the second-period value of information.  

\begin{proposition}
Assume that the programs $\max_{\one} \valeur^\Signal(\one)$ and
$\max_{\one} \valeur^{\Signal'}(\one)$ have unique\footnote{%
Unicity is for the sake of simplicity.
If the programs $\max_{\one} \valeur^\Signal(\one)$ and
$\max_{\one} \valeur^{\Signal'}(\one)$ do not have unique optimal solutions,
denote by $\argmax_{\one}  \valeur^\Signal(\one)$ and 
$\argmax_{\one}  \valeur^{\Signal'}(\one)$ the sets of maximizers. 
If $\info^\Signal(\one)$ is decreasing, the upper bounds of these sets
can be ranked: $ \sup \argmax_{\one}  \valeur^\Signal(\one)
\leq \sup \argmax_{\one}  \valeur^{\Signal'}(\one)$.
If $\info^\Signal(\one)$ is strictly decreasing, we obtain that 
$ \sup \argmax_{\one}  \valeur^\Signal(\one)
\leq \inf \argmax_{\one}  \valeur^{\Signal'}(\one)$ that is,
$  \bar{\one}^\Signal \leq \bar{\one}^{\Signal'}  $ for any 
$ \bar{\one}^\Signal \in \argmax_{\one}  \valeur^\Signal(\one)$ and
any $ \bar{\one}^{\Signal'} \in \argmax_{\one}  \valeur^{\Signal'}(\one)$.
}
 optimal solutions 
$\bar{\one}^\Signal$ and  $\bar{\one}^{\Signal'}$. 
Whenever the second-period value of the information is a 
decreasing function of the initial decision, namely
\begin{equation}
\info^{\Signal  \Signal'} :  \one \mapsto \valeur^\Signal(\one) - 
\valeur^{\Signal'}(\one) \mtext{ is decreasing,} 
\end{equation}
then
\begin{equation} 
 \bar{\one}^\Signal \leq \bar{\one}^{\Signal'} \; .
\end{equation}
\label{pr:main}
\end{proposition}

In the case where $\Signal'$ is constant (no information since 
$\sigma(\Signal')$ is the trivial $\sigma$-field 
$\{ \emptyset, \Om \}$),  the effect
of learning is  \emph{precautionary} in the sense that the optimal
initial decision is lower with information than without.

\subsection{Second-period value of information and Epstein functional}

We extend the approach of \citeauthor{Epstein1980:decision} 
in \citep*{Epstein1980:decision} and
\citep*{Jones-Ostroy:1984} to non necessarily finite sets.
We denote by $\calP(\RV)$ the Borel space of probability measures on  $\RV$,
with its Borel $\sigma$-field; the same holds for $\calP(\SIGNAL)$.

Following Epstein, let us define what we shall coin the 
\emph{Epstein functional} by\footnote{%
We denote by $\util(\one,\two,\cdot)$ the mapping $\RV \to \RR$ given by 
$\rv \mapsto \util(\one,\two,\rv)$.
}:
\begin{equation}
J(\one,\rho) \defegal 
\sup_{\two \in \dec(\one)}
\Expect{\rho}{ \util\big( \one,\two,\cdot \big) } =
\sup_{\two \in \dec(\one)} \int_\RV \util\big( \one,\two,\rv \big)
d\rho(\rv) \; , \quad \forall \rho \in \calP(\RV) \; .
\label{eq:Epstein_functional}
\end{equation}

Denote by $\nu$ and $\nu'$ the unconditional distributions 
of the signals $\Signal$ and $\Signal'$ on their image set $\SIGNAL$. 
The conditional distribution of $\Rv$ knowing $\Signal$ is a mapping
$\PP^{\Signal}_\Rv : \SIGNAL \to \calP(\RV)$.

The following result may be found in \citep*{Jones-Ostroy:1984}.
We give its proof in the general case of non necessarily finite sets.
\begin{proposition}
Assume that
\begin{enumerate}
\item for any pair of initial decisions $\one\upper \geq \one\llower$, 
$\rho \in \calP(\RV) \mapsto 
J(\one\upper,\rho) - J(\one\llower,\rho)$ is convex (resp. concave),
\item $\Signal$ is more informative than $\Signal'$.
\end{enumerate}
Then, the value 
$\info^{\Signal \Signal'}(\one) =
\valeur^\Signal(\one) - \valeur^{\Signal'}(\one)$
of substituting $\Signal$ for $\Signal'$, 
is increasing (resp. decreasing) with initial decision 
$\one \in \interval$.
Hence, Proposition~\ref{pr:main} applies.
\label{pr:Epstein_functional}
\end{proposition}

\begin{proof}
We have: 
\begin{eqnarray*} 
\valeur^\Signal(\one) &=& \Expect{\PP}{ \max_{\two \in \dec(\one)}
\condexpect{\PP}{ \util\big( \one,\two,\Rv \big) }{\Signal} } 
\mtext{ by definition~\eqref{eq:valuesignal}}\\
&=& \Expect{\nu}{ \max_{\two \in \dec(\one)}
\expect{\PP^{\Signal}_\Rv}{ \util\big( \one,\two,\cdot \big) } }  
\mtext{ by using the conditional distribution} \\
&=& \Expect{\nu}{ J(\one,\PP^{\Signal}_\Rv) }  
\mtext{ by~\eqref{eq:Epstein_functional}.} 
\end{eqnarray*}
Let $\one\upper \geq \one\llower$, and
suppose that  $\varphi(\rho)=J(\one\upper,\rho) - J(\one\llower,\rho)$ is
convex in $\rho \in \calP(\RV)$. 
We have
\begin{eqnarray*}
 \info^{\Signal \Signal'}(\one\upper) -
 \info^{\Signal \Signal'}(\one\llower) &=& 
\Expect{\nu}{ J(\one\upper,\PP^{\Signal}_\Rv) -
J(\one\llower,\PP^{\Signal}_\Rv) } - 
\Expect{\nu'}{ J(\one\upper,\PP^{\Signal'}_\Rv) -
J(\one\llower,\PP^{\Signal'}_\Rv) } \\ 
&=& 
\Expect{\nu}{ \varphi(\PP^{\Signal}_\Rv) } -
\Expect{\nu'}{ \varphi(\PP^{\Signal'}_\Rv) } \geq 0 
\end{eqnarray*}
since $\Signal$ is more informative than $\Signal'$.
Indeed, it is known (see \citep*{Dellach},
\citep*{Artstein-Wets:1993}, 
\citep*{Artstein1999:gains})
that if $\Signal$ is more informative than $\Signal'$, then
$ \expect{\nu}{\varphi(\PP^{\Signal}_\Rv)} \geq 
 \expect{\nu'}{\varphi(\PP^{\Signal'}_\Rv)}$ for all convex function
$\varphi$ on $\calP(\RV)$.\footnote{%
Notice that $\expect{\nu}{}$ denotes a mathematical expectation taken on
the probability space $(\SIGNAL,\tribu{\Signal},\nu)$.
}
The converse holds for concave.
\end{proof}

\section{Conditions on the primitive utility for the precautionary effect}
\label{sec:primitive}

\citeauthor{Epstein1980:decision}'s condition for the precautionary
effect in~\citep*{Epstein1980:decision} relies upon 
convexity (or concavity) of $ \rho \mapsto 
\frac{\partial J}{\partial \one}(\bar{\one}^\Signal,\rho)$.
In the same vein, \citeauthor{Jones-Ostroy:1984} rely upon 
convexity of the mapping $\rho \in \calP(\RV) \mapsto 
J(\one\upper,\rho) - J(\one\llower,\rho)$ in \citep*{Jones-Ostroy:1984}.
With the expression~\eqref{eq:Epstein_functional} of
$J(\one,\rho)$, it is not easy to see how this
relates to the primitive of the model, namely the utility $\util$. 

We shall proceed in two steps to characterize utility functions such
that $\rho \in \calP(\RV) \mapsto 
J(\one\upper,\rho) - J(\one\llower,\rho)$ is convex.
The difficulty comes from the fact that this latter function is the
difference of two convex functions, hence has no reason to be convex.
We shall, first, provide
a geometric characterization and, second, carry it functionally 
to the utility model by means of so-called \emph{support functions}.

\subsection{A geometric characterization}

The following characterization of when 
$\rho \in \calP(\RV) \mapsto J(\one\upper,\rho) - J(\one\llower,\rho)$ 
is convex relies upon the notion of sum of subsets of a vector space. 
Recall that, for any subsets $\Lambda_1$ and $\Lambda_2$ of a vector space, 
$\Lambda_1 + \Lambda_2 = \{ x_1 + x_2 \; , \, x_1 \in \Lambda_1
\mtext{ and } x_2 \in \Lambda_2 \}$ 
is their so called direct sum, or \emph{Minkowsky} sum.

Let us define, for any initial decision $\one \in \interval$,
\begin{equation}
 \Lambda^-(\one) \defegal \{ f : \RV \to \RR \mid
\mtext{ there exists } \two \in \dec(\one) \mtext{ such that }
f(\rv) \leq \util(\one,\two,\rv) \; , \quad \forall \rv \in \RV \}
\end{equation}
the set of maximal possible random rewards when the initial decision is
$\one$. 
Our first main result is the following. 

\begin{proposition}
Let $\one\upper > \one\llower$.
If 
there exists a subset $K$ 
of functions defined on  $\RV$ such that
\begin{equation}
\Lambda^-(\one\upper) = \Lambda^-(\one\llower) + K \; ,
\label{eq:geometric_characterization}
\end{equation}
then
$\rho \in \calP(\RV) \mapsto J(\one\upper,\rho) - J(\one\llower,\rho)$ is
convex.
Hence, the first hypothesis of Proposition~\ref{pr:Epstein_functional}
is satisfied.
\label{pr:sum}
\end{proposition}

\begin{proof}
 The proof comes from the observation that the 
Epstein functional $J(\one,\rho)$ is a so-called 
\emph{support function}\footnote{%
It is well known that an expected utility is convex in the prior.
It seems less noticed that the expected utility, seen as a
function of the prior, is the support
function of the set of payoffs indexed by actions.
}
as a function of its argument $\rho$. 
Recall that, to any set $ \Lambda $ of bounded measurable functions 
\( f : \RV \to \RR  \) is attached 
the support function \( \sigma_{\Lambda} \), defined on the 
Banach space of finite signed measures on  $\RV$, by:
\begin{equation}
\sigma_{\Lambda}(\rho) \defegal 
\sup_{\lambda \in \Lambda} \int_{\RV} \lambda(\rv) d\rho(\rv) \; .
\end{equation}
Indeed,
\eqref{eq:Epstein_functional} may be written as
\begin{equation}
 J(\one,\rho) = \sup_{\two \in \dec(\one)}
\Expect{\rho}{ \util\big( \one,\two,\cdot \big) } = 
\sup_{\lambda \in \Lambda(\one)} \Expect{\rho}{ \lambda } =
\sigma_{\Lambda(\one)}(\rho) \; , 
\end{equation}
where 
\begin{equation}
\Lambda(\one)\defegal \{ f : \RV \to \RR \mid
\mtext{ there exists } \two \in \dec(\one) \mtext{ such that }
f(\rv)= \util(\one,\two,\rv) \; , \quad \forall \rv \in \RV \} \; .
\label{eq:Lambda}
\end{equation}
Now, since $\rho$ belongs to cone of positive measures 
($\rho \in \calP(\RV)$), we also have
that $ J(\one,\rho) = \sigma_{\Lambda(\one)}(\rho) 
= \sigma_{\Lambda^-(\one)}(\rho) $, because
the polar cone of nonnegative functions on $\RV$ is the set of measures
\citep*[p.31,p.107]{Aubin:1982}.

Support functions have the nice property to transform a
Minkowsky sum into a sum of functions \citep*{Aubin:1982}:
\begin{equation}
\sigma_{\Lambda_1 + \Lambda_2} = \sigma_{\Lambda_1} + \sigma_{\Lambda_1}
\; . 
\end{equation}
Thus, whenever $\Lambda^-(\one\upper) = \Lambda^-(\one\llower) + K $, we
obtain that 
\begin{eqnarray*}
 J(\one\upper,\rho) - J(\one\llower,\rho) &=&  
\sigma_{\Lambda^-(\one\upper)}(\rho) - 
\sigma_{\Lambda^-(\one\llower)}(\rho) \\
&=& \sigma_{\Lambda^-(\one\llower) + K }(\rho) - 
\sigma_{\Lambda^-(\one\llower)}(\rho) \\
&=& \sigma_{\Lambda^-(\one\llower) }(\rho) + \sigma_{K }(\rho) - 
\sigma_{\Lambda^-(\one\llower)}(\rho) = \sigma_{K }(\rho) \; .
\end{eqnarray*}
Hence, $\rho \mapsto J(\one\upper,\rho) - J(\one\llower,\rho)
= \sigma_{K }(\rho)$ is convex since it is a support
function.\footnote{%
The condition $\Lambda^-(\one\upper) = \Lambda^-(\one\llower) + K $ is 
almost necessary for $J(\one\upper,\rho) - J(\one\llower,\rho)$ 
to be convex in $\rho \in \calP(\RV)$.
Indeed, \citep*[p.~92-93]{Hiriart-Urrurty-Lemarechal-I:1993} states that
the closed convex hull 
$\overline{co} (\sigma_{\Lambda_2} - \sigma_{\Lambda_1}) $ is equal to
$ \sigma_{\Lambda_2 \stackrel{\star}{-} \Lambda_1} $,
where the \emph{star-difference} 
$\Lambda_2 \stackrel{\star}{-} \Lambda_1 \defegal 
\{ x \mid x + \Lambda_1 \subset \Lambda_2 \} $.
Recall that the \emph{closed convex hull} $\overline{co} f$
of a function $f$ minorized by an affine function is the largest closed
convex function minorizing $f$.
}
\end{proof}

\subsection{Characterization of utility functions ensuring the precautionary effect}

We shall now show how the above geometric
characterization~\eqref{eq:geometric_characterization} translates
into a condition on the utility function $\util$. 

Consider two initial decisions $\one\upper > \one\llower$.
To any mapping $\phi : \TWO(\one\llower) \to \TWO(\one\upper) $
between second decision sets,
associate the following set of minimizers
\begin{equation}
 \dec_{\phi}(\one\upper,\one\llower,\rv) \defegal 
\argmin_{\two \in \dec(\one\llower)} \Big( 
\util(\one\upper,\phi(\two), \rv) - \util(\one\llower,\two, \rv) \Big) 
\; , \quad \forall \rv \in \RV \; .
\label{eq:minimizers}
\end{equation}
and
\begin{equation}
\dec_{\phi}(\one\upper,\one\llower) \defegal 
\bigcap_{\rv \in \RV} \dec_{\phi}(\one\upper,\one\llower,\rv) \; .
\end{equation}
When this latter set is not empty, there exists at least
one minimizer $\two \in \dec(\one\llower)$ of
$ \util(\one\upper,\phi(\two), \rv) - \util(\one\llower,\two, \rv) $
independent of the alea $\rv$.
Our second main result is the following. 

\begin{proposition}
Assume that
\begin{enumerate}
\item the set 
$\Phi(\one\upper,\one\llower) \defegal \{ \phi : \TWO(\one\llower) \to \TWO(\one\upper) \mid 
\dec_{\phi}(\one\upper,\one\llower) \not = \emptyset \} $ is not empty,
\item to any $\two\upper \in \TWO(\one\upper) $ 
can be associated at least one $\phi \in \Phi(\one\upper,\one\llower)$ 
and one $ \two\llower \in \dec_{\phi}(\one\upper,\one\llower) $ such that
$ \two\upper = \phi ( \two\llower ) $.
\end{enumerate}
Then there exists a subset $K$ of functions defined on 
$\RV$ such that  
$ \Lambda^-(\one\upper) = \Lambda^-(\one\llower) + K $.
Hence, the assumption of Proposition~\ref{pr:sum} is satisfied. 
\label{pr:minimizers}
\end{proposition}

\begin{proof}
Let us define the set of functions
$K = \cup_{\phi \in \Phi(\one\upper,\one\llower)} K_{\phi}$, where
\begin{equation}
 K_{\phi} = \{ \rv \in \RV \mapsto \util(\one\upper,\phi(\two), \rv) - 
\util(\one\llower,\two, \rv) \mtext{ for } \two \in 
\dec_{\phi}(\one\upper,\one\llower) \} \; . 
\end{equation}

\begin{enumerate}
 \item We first show the inclusion 
$\Lambda^-(\one\upper) \supset \Lambda^-(\one\llower) + K $.
Indeed, for any $\phi \in \Phi(\one\upper,\one\llower)$ 
and $\two \in \dec(\one\llower)$,
we have by definition of $K_{\phi}$ and $\dec_{\phi}(\one\upper,\one\llower)$:
\begin{equation*}
 k(\rv) \leq 
 \util(\one\upper,\phi(\two), \rv) - \util(\one\llower,\two, \rv) 
\; , \quad \forall \rv \in \RV \; , \quad \forall k \in K_{\phi} \; .
\end{equation*}
Hence, $ k(\rv) + \util(\one\llower,\two, \rv) \leq 
 \util(\one\upper,\phi(\two), \rv) $.
Thus, $\Lambda(\one\llower) + K \subset \Lambda^-(\one\upper) $, and
therefore $\Lambda^-(\one\llower) + K \subset \Lambda^-(\one\upper) $.

 \item We now show the reverse inclusion 
$\Lambda^-(\one\upper) \subset \Lambda^-(\one\llower) + K $.
Let $ \two\upper \in \dec(\one\upper) $. 
By assumption, there exist $\phi \in \Phi(\one\upper,\one\llower)$ and 
$ \two\llower \in \dec_{\phi}(\one\upper,\one\llower) $ such that
$ \two\upper = \phi ( \two\llower ) $.
We have that
\begin{equation}
\util(\one\upper,\two\upper,\rv) = 
\util(\one\upper, \phi(\two\llower),\rv) 
= \util(\one\llower,\two\llower, \rv) + 
\underbrace{\util(\one\upper, \phi(\two\llower),\rv) -
\util(\one\llower,\two\llower, \rv)}_{k(\rv) }
\end{equation}
where the function $k$ belongs to $K_{\phi} $ since 
$ \two\llower \in \dec_{\phi}(\one\upper,\one\llower) $.
\end{enumerate}

\end{proof}

The following Corollary provides practical conditions on the utility
function \( \util \) which ensure that the assumptions of 
Propositions~\ref{pr:minimizers} and ~\ref{pr:sum} are satisfied,
hence that the first hypothesis of Proposition~\ref{pr:Epstein_functional}
is satisfied. 

What is more, the first-order condition~\eqref{eq:first-order} 
has proximities with the second-order one in
\citep*{Salanie-Treich:2007}. This opens the way for comparison between
our approach and the invariance approach of
\citeauthor*{Salanie-Treich:2007}.

\begin{corollary}
Assume that the second decision variable $\two$ belongs to 
$\TWO=\RR^n$ and that 
the minimizers in~\eqref{eq:minimizers}
are characterized by the first-order optimality condition
\begin{equation}
\phi'(\two) \frac{\partial \util}{\partial \two} (\one\upper,\phi(\two), \rv) 
- \frac{\partial \util}{\partial \two} (\one\llower,\two, \rv) = 0
\; , \quad \forall \rv \in \RV \; .
\label{eq:first-order_optimality_condition}
\end{equation}
Suppose that, to any vector $\two\upper \in \dec(\one\upper) $ 
can be associated at least one vector
$ \two\llower \in \dec(\one\llower) $ and one 
square matrix $\Matrix \in \RR^{n\times n}$ such that
\begin{equation}
\Matrix \frac{\partial \util}{\partial \two} (\one\upper,\two\upper, \rv) 
- \frac{\partial \util}{\partial \two} (\one\llower,\two\llower, \rv) = 0
\; , \quad \forall \rv \in \RV \; .
\label{eq:first-order}
\end{equation}
If, in addition, we have 
$\two\upper + \Matrix(\two-\two\llower) \in \dec(\one\upper) $
for all $\two$ in a neighbourhood of $\two\llower $ 
in $\dec(\one\llower) $,\footnote{%
This condition is meaningless if $\two\upper$ belongs to the interior of 
$\dec(\one\llower) $. Hence this condition has to be verified only when an
irreversibility constraint bites.
}
then the assumptions of Proposition~\ref{pr:minimizers} are satisfied.
\label{cor:minimizers}
\end{corollary}

\begin{proof}
By assumption, to any $\two\upper \in \dec(\one\upper) $ 
can be associated at least one mapping defined by 
$\phi(\two)= \two\upper + \Matrix(\two-\two\llower)$ in a neighbourhood
of $\two\llower $ (and smoothly prolongated outside).
The  first-order optimality
condition~\eqref{eq:first-order_optimality_condition} 
attached to~\eqref{eq:minimizers} 
admits the solution $\two=\two\llower$ by~\eqref{eq:first-order}.

\end{proof}



\section{Analysis of examples in the literature}
\label{sec:Analysis of examples in the literature}

We shall examine a large body of the literature, and see that 
Corollary~\ref{cor:minimizers} applies in all cases and explains 
the precautionary effect.

At first, we shall assume that the second decision set $\dec(\one)=\TWO
$ does not depend on the 
initial decision $\one$, to concentrate on the precautionary effect and to try
to disentangle it from the irreversibility effect.
Second, we shall relax this assumption and attempt to point out the
impact that the second decision set $\dec(\one)$ indeed depends upon 
the initial decision $\one$.

In the sequel, we consider two initial decisions $\one\upper > \one\llower$.

\subsection{Additive separable preferences}

The case of additive separable preferences may be found in 
\citep*{Arrow.ea1974:environmental},
\citep*{Henry1974:investment},
\citep*{Epstein1980:decision},
\citep*{Freixas.ea1984:irreversibility},
\citep*{Fisher.ea1987:quasi},
\citep*{Hanemann1989:information} and is formalized by
\[
\util(\one,\two,\rv)= u(\one,\rv) + v(\two,\rv) \; .
\]
It can be seen that $\Lambda(\one)$ defined in \eqref{eq:Lambda} may be
written as 
\[
\Lambda(\one\upper)= \{ u(\one\upper,\rv) + 
v(\two,\rv) \, , \, \two \in \TWO \} =
u(\one\upper,\rv) - u(\one\llower,\rv) + 
\{ u(\one\llower,\rv) + v(\two,\rv) \, , \, \two \in \TWO \} 
= K + \Lambda(\one\llower)
\]
where $K$ is the singleton $\{ u(\one\upper,\rv) - u(\one\llower,\rv) \}$. 
Hence \( J(\one\upper,\rho) - J(\one\llower,\rho) =
\sigma_{u(\one\upper,\rv) - u(\one\llower,\rv)}(\rho) 
= \int_{\RV} \big( u(\one\upper,\rv) - u(\one\llower,\rv) \big) 
d\rho(\rv) \) is linear in the
prior $\rho$, hence both concave and convex. 
Thus, the precautionary effect holds true in a strong sense since the
initial optimal decision does not depend on the amount of information
the agent expects to receive. 

The above analysis is confirmed by the 
first-order optimality condition~\eqref{eq:first-order} which is 
\begin{equation*}
\Matrix \frac{\partial v}{\partial \two}(\two\upper, \rv) =
 \frac{\partial v}{\partial \two}(\two\llower, \rv)  
\; , \quad \forall \rv \in \RV \; .
\end{equation*}
Since $\Matrix = \mbox{\rm Id}_{n\times n}$ and
$ \two\llower=\two\upper $ are solutions, 
the precautionary effect holds true.

However, when $\dec(\one)$ indeed depends upon $\one$,
we no longer have that 
$\Lambda(\one\upper)= K + \Lambda(\one\llower)$.
We also observe that
the additional conditions of Corollary~\ref{cor:minimizers}
related to the irreversibility constraints, namely  
$ \dec(\one\upper) \subset \dec(\one\llower) $ and
$\two\upper \in \dec(\one\upper) $
for all $\two$ in a neighborhood of $\two\upper \in \dec(\one\llower)$,
are not generaly satisfied and this 
may prevent the precautionary effect to hold true.

\subsection{Risk neutral preferences}

Examples in \citep*{Epstein1980:decision,Ulph.ea1997:global} present the
general structure 
\[
\util(\one,\two,\rv) = u(\one,\two) + v(\one,\two) \rv 
= u(\one,\two) + \sum_{i=1}^p v_i(\one,\two) \rv_i  \; . 
\]
The first-order optimality condition~\eqref{eq:first-order} is
\begin{equation*}
\left\{ \begin{array}{rcl}
\displaystyle 
\Matrix \frac{\partial u}{\partial \two}(\one\upper,\two\upper) &=&
\displaystyle 
\frac{\partial u}{\partial \two} (\one\llower,\two\llower) \\[4mm]
\displaystyle 
\Matrix \frac{\partial v_i}{\partial \two}(\one\upper,\two\upper) &=& 
\displaystyle 
\frac{\partial v_i}{\partial \two}(\one\llower,\two\llower) 
\, , \quad i=1,\ldots, p \; .
\end{array} \right.
\end{equation*}
This is a system of $n + np $ equations with 
$n+ n^2$ unknown $(\Matrix ,\two\llower)$.
Thus, when the dimension $p$ of the noise is less than the dimension $n$  
of the second decision variable, the precautionary effect generally
holds true.

\subsection{Risk averse preferences}

We shall examine three models where preferences exhibit risk aversion.

\subsubsection*{A consumption-savings problem \citep*{Epstein1980:decision}}

A two-periods consumption-savings problem is modelled by 
\[
\util(\one,\two, \rv) = u_1(w-\one) + \beta u_2(r\one - \two) 
+ \beta^2 u_3(\two\rv) \; ,
\]
with savings $\one, \two$, and irreversibility constraint 
$\dec(\one)=[0, r \one] $. 

The first-order optimality condition~\eqref{eq:first-order} is
\[
\Matrix \beta \rv u_3'(\two\upper \rv) - \beta \rv u_3'(\two\llower \rv) =
\Matrix u_2'(r\one\upper - \two\upper) - u_2'(r\one\llower - \two\llower)
\; , \quad \forall \rv \in \RV \; .
\]
If there exists a solution  $(\Matrix ,\two\llower)$, this 
implies that there must exist constants $\alpha$, $\gamma$ and $\delta$
such that $u_3'$ satisfies an equation of the form 
\[
\rv u_3'(\alpha \rv) = \gamma \rv u_3'(\rv) + \delta 
\; , \quad \forall \rv \in \RV \; .
\]
A candidate is $u_3'(\rv)=\rv^{-\gamma}$, yielding a solution
$\Matrix =\gamma $ and $ \two\llower = \alpha \two\upper $,
with the compatibility condition 
$ \gamma u_2'(r\one\upper - \two\upper) 
- u_2'(r\one\llower - \alpha \two\upper ) + \delta = 0 $.
Hence, the precautionary effect holds true. 

\subsubsection*{Global warming and emissions 
\citep*{Gollier.ea2000:scientific} }

With pollution emissions $\one, \two$, a two-periods model with benefits
and costs of emitting pollutions is modelled by 
\[
\util(\one,\two, \rv) = 
u(\one) + v\big( \two - \rv(\one+\two) \big) \; .
\]
The first-order optimality condition~\eqref{eq:first-order} is
\begin{equation*}
\Matrix v' \big( \two\upper - \rv(\one\upper+\two\upper) \big) =
v' \big( \two\llower - \rv(\one\llower+\two\llower) \big) 
\; , \quad \forall \rv \in \RV \; .
\end{equation*}
If there exists a solution  $(\Matrix ,\two\llower)$, this 
implies that there must exist constants $\alpha$, $\beta$ and $\Matrix $
such that $v'$ satisfies an equation of the form 
\[
 v'(\alpha \rv + \beta) = \Matrix v'(\rv) 
\; , \quad \forall \rv \in \RV \; .
\]
In this case, 
$\two\llower = \two\upper \one\llower / \one\upper $.
The utility $v(x)= \frac{\gamma}{1-\gamma} 
\left[ \eta + \frac{x}{\gamma} \right]^{1-\gamma} $ in 
\citep*{Gollier.ea2000:scientific} precisely satisfies
$ v'(\alpha \rv + \gamma \eta (\alpha-1)) = \alpha^{-\gamma} v'(\rv) $,
which explains why the precautionary effect holds true. 

\subsubsection*{Eating a cake with unknown size \citep*{Eeckhoudt-Gollier-Treich:2005}}

The following model from \citep*{Eeckhoudt-Gollier-Treich:2005} 
is qualified of the problem of ``eating a cake with unknown size''
in \citep*{Salanie-Treich:2007}:
\[
\util(\one,\two, \rv) = u(\one) + v(\two) + w(\rv -\one-\two) 
\; . 
\]
The first-order optimality condition~\eqref{eq:first-order} is
\begin{equation*}
\Matrix v'(\two\upper) - v'(\two\llower) = 
\Matrix w'(\rv-(\one\upper+\two\upper)) - 
w'(\rv-(\one\llower+\two\llower)) \; .
\end{equation*}
If there exists a solution  $(\Matrix ,\two\llower)$, this 
implies that there must exist constants $\beta$, $\kappa$ and $\Matrix $
such that $w'$ satisfies an equation of
the form $w'(\rv + \beta) = \Matrix v'(\rv) + \kappa $. 
Then, we find that
$\beta + \one\upper + \two\upper = \one\llower + \two\llower $
and that
$\Matrix v'(\two\upper) - v'(\two\llower) + \kappa = 0$.
Thus, $\beta$, $\kappa$ and $\Matrix $ must satisfy some compatibility
constraint, so that the precautionary effect holds true.

\section{Conclusion}

We have provided general conditions for the precautionary effect 
to hold true, including a condition that bears directly on the primitive
utility of the economic model.
We have examined a large body of the literature, and seen how operative
is this condition to explain the precautionary effect.
Preferences yielding the precautionary effect for all signals appear
to belong to a restricted class.
This is related to the strong conditions, that we provide,
needed to have a difference of maximal payoffs exhibit convexity in the
prior. The connexion with the invariance approach of
\citep*{Salanie-Treich:2007} deserves to be studied and clarified. 

\bigskip

\textbf{Acknowledgements.} 
We want to thank Christian Gollier, François Salanié and Nicolas Treich 
for their invitation to the seminar at LERNA (Toulouse School of
Economics) on 19th January 2009, and for the subsequent discussions.

\end{document}